\documentclass[aps,pra,twocolumn,groupedaddress]{revtex4}

\usepackage[latin1]{inputenc}
\usepackage[T1]{fontenc}
\usepackage[english]{babel}
\usepackage[dvips]{graphicx}
\usepackage{enumerate,amsthm,amsmath,amssymb,color,graphicx,bbm,float}
\usepackage{natbib}

\makeatletter
\newtheorem*{rep@theorem}{\rep@title}
\newcommand{\newreptheorem}[2]{%
\newenvironment{rep#1}[1]{%
 \def\rep@title{#2 \ref{##1}}%
 \begin{rep@theorem}}%
 {\end{rep@theorem}}}
\makeatother

\newtheorem{theor}{Theorem} 
\newtheorem{theo}{Theorem} [section]
\newtheorem{lemma}[theo]{Lemma}

\newreptheorem{lemma}{Lemma}

\usepackage{epstopdf,hyperref}

\newcommand{\be}{\begin{equation}}
\newcommand{\ee}{\end{equation}}

\usepackage{mathtools}
\def\multiset#1#2{\ensuremath{\left(\kern-.3em\left(\genfrac{}{}{0pt}{}{#1}{#2}\right)\kern-.3em\right)}}

\begin{document}

\author{Anthony Leverrier$^1$, Ra\'{u}l Garc{\'i}a-Patr\'{o}n$^2$, Renato Renner$^1$, Nicolas J. Cerf$^3$}
\affiliation{$^1$Institute for Theoretical Physics, ETH Zurich, 8093 Zurich, Switzerland }
\affiliation{$^2$Max-Planck Institut fur Quantenoptik, Hans-Kopfermann Str. 1, D-85748 Garching, Germany}
\affiliation{$^3$Quantum Information and Communication, Ecole Polytechnique de Bruxelles, CP 165, Université Libre de Bruxelles, 1050 Bruxelles, Belgium}

\title{Security of continuous-variable quantum key distribution against general attacks}

\date{\today}

\begin{abstract}
We prove the security of Gaussian continuous-variable quantum key distribution against arbitrary attacks in the finite-size regime. The novelty of our proof is to consider symmetries of quantum key distribution in phase space in order to show that, to good approximation, the Hilbert space of interest can be considered to be finite-dimensional, thereby allowing for the use of the postselection technique  introduced by Christandl, Koenig and Renner (\emph{Phys.\ Rev.\ Lett.}\ 102, 020504 (2009)). Our result greatly improves on previous work based on the de Finetti theorem which could not provide security for realistic, finite-size, implementations.
\end{abstract}

\maketitle

Quantum key distribution (QKD), the art of generating a secret key among distant parties in an untrusted environment, is certainly the most studied quantum cryptographic primitive. 
Since the seminal papers of Bennett and Brassard \cite{bb84} and Ekert \cite{eke91}, considerable progress has been made in terms of security analysis \cite{SBC08}. 
Security against arbitrary attacks has been proven for several protocols, even in the realistic finite-size regime. This is quite remarkable because of the very large number of possible attacks against which security needs to be guaranteed. 
Security proofs generally circumvent this problem by using the natural permutation invariance of most QKD protocols which allows to restrict the analysis to the much smaller class of \emph{collective} attacks, where the eavesdropper interacts independently and identically with every communication signal. 
In an entanglement-based description of QKD, this amounts to assume that the joint state $\rho_{A^nB^n}$ that the two legitimate parties, Alice and Bob, hold after the initial distribution of entanglement, has an identical and independently distributed (i.i.d.) structure $\rho_{A^nB^n}=\sigma_{AB}^{\otimes n}$, where  $n$ is the number of quantum signals exchanged during the protocol.

One usually achieves this reduction from general to collective (i.i.d.) attacks thanks to either de Finetti-type theorems \cite{ren07} or the postselection technique \cite{CKR09}. Unfortunately, these tools cannot be directly applied to continuous-variable (CV) protocols because they require the dimension of the Hilbert space to be finite (and small compared to $n$). However, by prepending a suitable energy test to the protocol, it is still possible to use a specific variant of the de Finetti theorem and derive the security of CV protocols, but only for impractically large values of $n$ \cite{RC09}.
Here, we wish to improve the analysis of \cite{RC09} to prove the security of continuous-variable QKD in the realistic finite-size scenario.

The specificity of CV protocols is that the detection consists of (homodyne or heterodyne) measurements of the light-field quadratures (see Ref.\ \cite{WPG12} for a review). From an experimental point of view, they present many advantages over discrete-variable protocols. Most importantly, they can be implemented with standard telecom components and are compatible with Wavelength Division Multiplexing \cite{QZQ10}, which is an important advantage when integrating QKD into real-world telecommunication networks.
Moreover, quadrature measurements do not require any photon counters and higher repetition rates can be achieved.
Distribution of secret keys over long distances (more than 80 km) is currently achievable \cite{JKL12}, making CV protocols competitive with respect to their discrete-variable counterparts. 
Their security analysis, however, is technically challenging due to the infinite-dimensional nature of the relevant Hilbert space. 

Among CV protocols, the so-called Gaussian ones are the most popular ones, primarily due to their experimental simplicity. 
In a prepare-and-measure scheme, one party, Alice, prepares coherent or squeezed states with a Gaussian modulation and sends them to a receiver, Bob, who applies a homodyne or heterodyne measurement.
In the equivalent entanglement-based scheme, Alice prepares an entangled two-mode squeezed vacuum state (the continuous-variable equivalent of the Bell pair), keeping one mode and sending the other one to Bob through the quantum channel. 
Then both parties measure their respective mode with either a homodyne or heterodyne detection, obtaining two strings of correlated real-valued data. Finally, Alice and Bob extract a secret key through information reconciliation and privacy amplification. 

The security of Gaussian protocols in the aymptotic regime is rather well understood: de Finetti's theorem guarantees that collective attacks are optimal \cite{RC09} and Gaussian attacks are known to be optimal among collective ones \cite{GC06,NGA06}. 
On the other hand, their security in the much more relevant  finite-size regime is less clear, due to the difficulty of the reduction from general attacks to the i.i.d. scenario. 
Currently, two results in this direction are known for CV protocols, either based on a de Finetti theorem as stated above or on an uncertainty relation. 
The de Finetti approach \cite{RC09} is unsuitable in practical scenarios because $n$, the required number of signals exchanged during the protocol, is too large.  
The second approach, using an entropic uncertainty inequality \cite{FFB11}, works for more reasonable values of $n$ but unfortunately does not converge towards the asymptotic key rate secure against collective attacks in the limit of infinitely many signals. Consequently, the tolerated losses are quite low,  corresponding to a few hundred meters only. 

In the remainder of this Letter, we first explain how to modify a protocol secure against collective attacks by the addition of an initial test in order to enforce a certain property of the entangled state, namely a low single-mode photon number. Then, we apply these ideas to the specific case of a Gaussian protocol where Bob performs heterodyne measurements and establish its security against arbitrary attacks.

{\bf Main result}.--- In this paper, we give the first security proof of CV QKD against general attacks, which guarantees a secret key rate for realistic experimental regimes, in terms of losses and noise. 
As in \cite{RC09,FFB11}, this is achieved by prepending an initial test to a protocol already proven secure against collective attacks. The purpose of the test is to verify that the quantum state shared by Alice and Bob is well-approximated by a state living in a reasonably small dimensional Hilbert space. Then, one can use the postselection technique \cite{CKR09} which shows roughly that if a (permutation-invariant) protocol is $\epsilon$-secure against collective attacks, then it is $\tilde{\epsilon}$-secure against general attacks with $\tilde{\epsilon} = \epsilon \times \mathrm{poly}(n)$.

Our result improves that of Ref.\ \cite{RC09} for two reasons. First the postselection technique guarantees much better bounds than the approach based on a de Finetti theorem when reducing general to collective attacks \cite{SLS10}. Moreover, and this is in fact the main technical contribution of the present work, we exploit specific symmetries of the CV QKD protocol in phase space instead of the usual and less powerful permutation symmetry. We therefore obtain a very tight bound on the effective dimension of the quantum state. More precisely, the QKD protocol is invariant if Alice and Bob process their respective modes with global conjugate passive linear transformations of their $n$ modes before performing their measurements. This ``rotational-symmetry'' in phase space is better suited to analyze CV protocols \cite{LKG09}, allowing to precisely bound the effective number of photons per mode from the results of random quadrature measurements. This is in stark contrast with Ref.\ \cite{RC09} where the test only exploited the permutation symmetry of the protocol.

{\bf QKD protocols and their security}.--- A QKD protocol is a CP map from the infinite-dimensional Hilbert space $(\mathcal{H}_A\otimes \mathcal{H}_B)^{\otimes n}$, corresponding to the initially distributed entanglement,  to the set of pairs $(S_A,S_B)$ of $l$-bit strings (Alice and Bob's final keys, respectively) and $C$, a transcript of the classical communication.
In order to assess the security of a given QKD protocol $\mathcal{E}$ in a composable framework, one compares it with an ideal protocol \cite{MKR09}. Such an ideal protocol $\mathcal{F}$ can be constructed (at least in principle) by concatenating the protocol with a map $\mathcal{S}$ taking $(S_A,S_B,C)$ as input and outputting the triplet $(S,S,C)$ where the string $S$ is a perfect secret key (uniformly distributed and unknown to Eve) with the same length as $S_A$, that is $\mathcal{F} = \mathcal{S}\circ\mathcal{E}$. Then, a protocol will be called \emph{$\epsilon$-secure} if the advantage in distinguishing it from an ideal version is not larger than $\epsilon$. This advantage is quantified by (one half of) the  diamond norm defined by
\begin{equation}
||\mathcal{E} - \mathcal{F}||_\diamond := \sup_{\rho_{ABE} } \left\|(\mathcal{E}-\mathcal{F})\otimes \mathrm{id}_\mathcal{K} (\rho_{ABE})\right\|_1,
\end{equation}
where the supremum is taken over $(\mathcal{H}_A\otimes \mathcal{H}_B)^{\otimes (n+k)} \otimes \mathcal{K}$ for any auxiliary system $\mathcal{K}$. 

{\bf Prepending a test}.--- Our main technical result is a reduction of the security against general attacks to that against collective attacks, for which security has already been proved in earlier work. 
Let us therefore suppose that our CV QKD protocol of interest, $\mathcal{E}_0$, is secure against collective attacks. We will slightly modify it by prepending an initial test $\mathcal{T}$. More precisely, $\mathcal{T}$ is a CP map taking a state in a slightly larger Hilbert space, $(\mathcal{H}_A\otimes \mathcal{H}_B)^{\otimes (n+k)}$, measuring $k$ randomly chosen modes (identical for Alice and Bob) and comparing the measurement outcome to a value fixed in advance. 
The test succeeds if this norm is small, meaning that the global state is compatible with a state containing only a low number of photons per mode, that is a state well-described in a low dimensional Hilbert space, which leads to better bounds when using the post-selection technique. 
Depending on the outcome, either the whole protocol aborts, or one applies $\mathcal{E}_0$ on the $n$ remaining modes. 
A more precise description is provided as part of the ``heterodyne protocol'' below. 

For our purpose, it is crucial that the test is feasible in practice. This is the case here since it only involves $k$ additional homodyne (or heterodyne) measurements compared to the original scheme $\mathcal{E}_0$, with $k$ much smaller than $n$, as well as applying some classical post processing to Alice and Bob's data. 

In order to establish that the protocol $\mathcal{E}:=\mathcal{E}_0\circ \mathcal{T}$ is $\epsilon$-secure against arbitrary attacks, one needs to bound $||\mathcal{E}-\mathcal{F}||_\diamond$. 
The postselection theorem \cite{CKR09} allows one to bound the diamond norm between such maps by simply considering i.i.d.\ states (i.e. the equivalent of collective attacks), but only when the maps act on finite dimensional spaces. 
We address this issue by introducing another CP map $\mathcal{P}$ which projects a state in $(\mathcal{H}_A\otimes \mathcal{H}_B)^{\otimes n}$ onto a low-dimensional Hilbert space $(\overline{\mathcal{H}}_A\otimes \overline{\mathcal{H}}_B)^{\otimes n}$ where $\overline{\mathcal{H}}_A := \mathrm{Span}(|0\rangle, |1\rangle, \cdots, |d_A-1\rangle)$ and $\overline{\mathcal{H}}_B := \mathrm{Span}(|0\rangle, |1\rangle, \cdots, |d_B-1\rangle)$ are respectively a $d_A$ and a $d_B$-dimensional subspace of the Fock spaces $\mathcal{H}_A$ and $\mathcal{H}_B$.
We define (virtual) protocols $\tilde{\mathcal{E}}:=  \mathcal{E}_0 \circ \mathcal{P} \circ \mathcal{T}$ and $\tilde{\mathcal{F}}:=  \mathcal{S} \circ \tilde{\mathcal{E}}$. 
The security of the protocol $\mathcal{E}$ is then a consequence of the following derivation:
\begin{eqnarray}
 ||\mathcal{E} - \mathcal{F}||_\diamond &\leq &  ||\tilde{\mathcal{E}} - \tilde{\mathcal{F}}||_\diamond  +  ||\mathcal{E} - \tilde{\mathcal{E}}||_\diamond+ ||\mathcal{F} - \tilde{\mathcal{F}}||_\diamond \nonumber\\
&\leq &  ||\tilde{\mathcal{E}} - \tilde{\mathcal{F}}||_\diamond  +  ||\mathcal{E}_0  \circ (\mathrm{id}- \mathcal{P}) \circ \mathcal{T}||_\diamond \nonumber \\
&&+||\mathcal{F}_0  \circ (\mathrm{id}- \mathcal{P}) \circ \mathcal{T}||_\diamond \nonumber\\
&\leq &  ||\tilde{\mathcal{E}} - \tilde{\mathcal{F}}||_\diamond  + 2 || (\mathrm{id}- \mathcal{P}) \circ \mathcal{T}||_\diamond ,
\end{eqnarray} 
where we used the triangle inequality and the fact that the CP maps $\mathcal{E}_0$ and $\mathcal{F}_0$ cannot increase the diamond norm. 
The first term can be bounded thanks to the postselection theorem because $\tilde{\mathcal{E}}$ and $\tilde{\mathcal{F}}$ are finite dimensional, and it can be made arbitrary small at the price of reducing slightly the key rate. The second term can be bounded thanks to the following theorem for which we give a proof sketch for the ``heterodyne protocol'' below (and a full proof in Appendix \ref{proof_hetero}).
\begin{theor}(Informal.)
\label{informal}
For any rotationally-invariant state $\rho_{ABE} \in (\mathcal{H}_A\otimes \mathcal{H}_B)^{\otimes (n+k)} \otimes \mathcal{K}$,
\begin{equation}
||\left(\mathrm{id}_{\mathcal{H}^{\otimes n}} - \mathcal{P} \right) \circ \mathcal{T} \otimes \mathrm{id}_\mathcal{K} (\rho_{ABE}) ||_1 \leq \epsilon,
\end{equation}
where $\epsilon$ is a function of $k, n$, the dimensions $d_A$ and $d_B$ for the projection $\mathcal{P}$ and the value of the threshold in the test $\mathcal{T}$. 
\end{theor}

{\bf Description for the protocol with heterodyne detection}.--- Let us now consider a specific example of a QKD protocol $\mathcal{E}_0$. For the sake of clarity, we choose (arguably) the simplest one \cite{WLB04}. In the prepare-and-measure version of the protocol, Alice prepares $n$ coherent states which are modulated with a Gaussian distribution, and sent through the quantum channel. In the equivalent entangled version of the protocol, for which we prove security, Alice prepares $n$ two-mode squeezed vacuum states, measures one mode of each state with a heterodyne detection and sends the other one through the quantum channel. Bob then performs a heterodyne measurement of the states he receives. This means that he measures both quadratures $q$ and $p$ for each mode. This is achieved by sending the modes on a balanced beamsplitter and measuring the $q$ quadrature for one output mode and the $p$ quadrature for the other one. At the end of this process, Alice and Bob have access to two correlated vectors in $\mathbbm{R}^{2n}$, $\vec{x}_A$ for Alice and $\vec{x}_B$ for Bob. Then, they perform the reconciliation procedure \cite{JKL11} in order to extract a common string, and finally privacy amplification \cite{ren05} to distill their final secret keys, $S_A$ and $S_B$, respectively. 

This protocol is invariant under the action of conjugate passive symplectic operations (beamsplitters and phase shifts) because these correspond to some orthogonal transformation $R\in O(2n)$ of the quadratures in phase space. Specifically, if such an operation is applied, then Alice and Bob's vectors become $R \vec{x}_A$ and $R^T \vec{x}_B$ (see Appendix \ref{proof_symmetry} for details), meaning that the effect of the beamsplitters and phase shifts can be undone by simply applying the inverse rotation on the classical data. 

We assume in the following that the protocol $\mathcal{E}_0$ is secure against collective attacks, in the sense that for any pure state $\rho_{ABE} \in \mathcal{H}_A\otimes \mathcal{H}_B \otimes \mathcal{H}_E$ where $ \mathcal{H}_E \cong  \mathcal{H}_A\otimes \mathcal{H}_B$, the quantity $\left\| (\mathcal{E}_0-\mathcal{F}_0)\otimes \mathrm{id}_\mathcal{K} \left(\rho_{ABE}^{\otimes n}\right) \right\|_1$ can be made exponentially small in $n$, say $2^{-c\delta^2 n}$, at the price of reducing the secret key rate by an arbitrary small fraction $\delta$ compared to the asymptotic optimal rate, for some constant $c>0$. We note that despite being proven secure against collective attacks in the asymptotic limit  \cite{GC06,NGA06,LG10b}, the security of $\mathcal{E}_0$ for finite size attacks is not yet completely understood in the sense that the precise values of $c$ and $\delta$ are not currently known: this is due to the difficulty of estimating a covariance matrix in the finite-size regime (see \cite{LGG10}).

As we mentioned above, we will prove the security of a slightly modified protocol, noted $\mathcal{E}$ which starts with $n+k$ modes (instead of $n$ in the case of $\mathcal{E}_0$), $k$ of which being used to conduct a test $\mathcal{T}$. If the test passes, corresponding roughly to a scenario where the state does not contain too many photons, then Alice and Bob proceed with the protocol $\mathcal{E}_0$, otherwise they abort. 
The test $\mathcal{T}$ is in fact only applied to Bob's classical data. Indeed, we assume here that Alice prepares her state in a trusted environment meaning that her reduced state is an $(n+k)$-modal thermal state. Note that one could easily remove this assumption and also apply $\mathcal{T}$ to Alice's state. 

The test consists in first choosing a random rotation $R$ in $\mathbbm{R}^{2(n+k)}$ (with the appropriate measure) and applying it to the $2(n+k)$-dimensional vector corresponding to Bob's measurement outcomes (as well as to Alice's vector). 
Let us denote by $q_1, p_1, q_2, p_2, \cdots, q_k, p_k$ the first $2k$ coordinates of Bob's rotated vector and define the variable $Y_k :=  \sum_{i=1}^k (q_i^2 + p_i^2)$. The coordinates correspond to heterodyne measurements of $k$ modes of $\rho_B^{n+k}$ after being processed through an appropriate network of beamsplitters and phase-shifts (see Appendix \ref{proof_symmetry}). 
The test $\mathcal{T}$ is characterized by 2 parameters: a positive number $Y_\mathrm{test}$ and $k$. The test passes if $Y_k \leq Y_\mathrm{test}$ and fails otherwise. 
More precisely, because the test commutes with the measurement, it can equivalently be seen as  a map from $(\mathcal{H}_A \otimes \mathcal{H}_B)^{\otimes (n+k)}$ to $(\mathcal{H}_A \otimes \mathcal{H}_B)^{\otimes n}$ (plus an additional bit encoding whether the test passed or not) that returns the $n$ remaining modes when it passes and an ``abort'' state when it fails.

It is also useful to describe the CP map $\mathcal{P}$ characterized by three numbers, $n$, and the local dimensions $d_A$ and $d_B$. It corresponds to the binary outcome measurement in $(\mathcal{H}_A \otimes \mathcal{H}_B)^{\otimes n}$ described by the POVM $ \{ P_A^{\otimes n} \otimes P_B^{\otimes n}, \mathbbm{1} -  P_A^{\otimes n} \otimes P_B^{\otimes n} \}$ where $P_A$ and $P_B$ are the single-mode projectors on $\mathcal{H}_A$ and $\mathcal{H}_B$, respectively, defined as
$P_{A/B} = |0\rangle\!\langle 0| + |1\rangle\!\langle 1| + \cdots + |d_{A/B}-1\rangle\!\langle d_{A/B}-1|.$

In order to establish Theorem \ref{informal}, we will bound the probability $p_\mathrm{bad}$ of the following bad event: "the state passes the test and the projection onto $ P_A^{\otimes n} \otimes P_B^{\otimes n}$ fails" for some initial state $\rho_{AB}^n \in  (\mathcal{H}_A \otimes \mathcal{H}_B)^{\otimes n}$. Let us note $\tilde{\rho}_{AB}^n$ the unnormalized state after the test when it passed; the probability of passing the test is simply $p_\mathrm{test} = \mathrm{tr} \tilde{\rho}_{AB}^n$ and $p_\mathrm{bad} = \mathrm{tr}\left[ (1-\mathcal{P})\circ \mathcal{T} \left(\rho_{AB}^{n+k}\right)\right]$.
One can bound $p_\mathrm{bad}$ in the following way:
\begin{eqnarray}
p_\mathrm{bad} &=&\mathrm{tr}\left(\mathrm{id}_{AB}-  P_A^{\otimes n} \otimes P_B^{\otimes n}\right)\tilde{\rho}_{AB}^n \nonumber\\
& \leq & \mathrm{tr} \left[\left( \mathrm{id}_A- P_A^{\otimes n} \right) \tilde{\rho}_{A}^n\right]+  \mathrm{tr} \left[\left( \mathrm{id}_B- P_B^{\otimes n} \right) \tilde{\rho}_{B}^n\right] \nonumber\\
& \leq & \mathrm{tr} \left[\left( \mathrm{id}_A- P_A^{\otimes n} \right) {\rho}_{A}^n\right]+  \mathrm{tr} \left[\left( \mathrm{id}_B- P_B^{\otimes n} \right) \tilde{\rho}_{B}^n\right] \label{bound_bob}
\end{eqnarray}
where we used the union bound and the fact that Alice does apply the test. 
The first term is easy to compute because the state of Alice, a multimode thermal state, is well known: $\rho_A^n =  \rho_\mathrm{thermal}^{\otimes n}$ with $\rho_\mathrm{thermal} = \sum_{k=0}^\infty \frac{\lambda^k}{(1+\lambda)^{k+1}} |k\rangle\!\langle k|$ for a state with $\lambda$ photons per mode. The value of $\lambda$ is a parameter of the protocol and should be optimized given the expected characteristics of the quantum channel. 
The union bound gives
\begin{equation*}
1- \mathrm{tr} \left(P_A^{\otimes n} \rho_A^n \right) \leq n (1-  \mathrm{tr} \left(P_A  \rho_\mathrm{thermal} \right))=n\left(  \frac{\lambda}{1+\lambda}\right)^{d_A}. 
\end{equation*}
In particular, choosing $d_A =\frac{ \log (n/\epsilon_A) }{\log (1+1/\lambda)}$ for the dimension of Alice's Hilbert space leads to $1 - \mathrm{tr}\left(P_A^{\otimes n}  \rho_{A}^n\right)  \leq \epsilon_A$. 

Bounding the second term in Eq.\ (\ref{bound_bob}) is much trickier because one cannot assume that Bob's state $\rho_B^n$ is Gaussian or that it even has an i.i.d. structure. This is because it corresponds to the output of the unknown quantum channel controlled by Eve. Here, we will make use of the specific symmetries of the QKD protocol in phase space in order to simplify greatly the problem. 
In general, most protocols are invariant under permutations of the subsystems of Alice and Bob. This means that the state $\rho_{AB}^n$ (and therefore also $\rho_B^n$) can be assumed to display this invariance. However, CVQKD protocols such as the one considered here respect a much stronger symmetry: they are invariant when Alice and Bob apply to their respective $(n+k)$ modes conjugate passive linear transformations, implemented by any network of beamsplitters and phase shifts \cite{LKG09,LG10b} (see Appendix \ref{proof_symmetry} for details). Here, it is crucial that the test $\mathcal{T}$ respects the symmetry, and this can be enforced \emph{at the level of classical data} by the choice of the random subspace $T$ of $\mathbbm{R}^{2(n+k)}$ (as explained before). 

Thanks to this symmetry, one can assume that the state $\rho_B^{n+k}$ of Bob (before applying the test $\mathcal{T}$) is rotationally invariant, that is, left invariant under the action of any network of passive linear operations on their $n+k$ modes. Such states were already studied in Ref.\ \cite{LC09} where a de Finetti theorem was established: if sufficiently many modes of $\rho_B^{n+k}$ are traced out, then the remaining state is close to a mixture of thermal states. Intuitively, one then expects that the second term of Eq.\ \ref{bound_bob} behaves like the first one, and this is what we prove rigorously. 
Before we explain how to bound $\mathrm{tr}\left(P_B^{\otimes n}  \tilde{\rho}_{B}^n\right)$, we recall two useful properties of states, such as $\rho_B^{n+k}$, which are rotationally invariant \cite{LC09}. First, these states are mixtures of generalized $(n+k)$-mode Fock states $\sigma_p^{n+k} := 1 / {n+k+p-1 \choose p} \sum_{p_1 + \cdots + p_m=p} |p_1, p_2, \cdots, p_m\rangle \!\langle p_1, p_2, \cdots, p_m| $, where $|p_1, \cdots, p_m \rangle$ is the product of Fock states with $p_1$ photons in the first mode, $p_2$ photons in the second mode, etc, and the sum is taken over all states with a total number of $p$ photons in $n+k$ modes. This means that there exist $\lambda_0 \geq 0 , \lambda_1 \geq 0, \cdots$ such that $\rho_B^{n+k} = \sum_{p=0}^\infty \lambda_p \sigma_p^{n+k}$.
The second useful property is that the Wigner function $W(q_1, p_1, \cdots, q_{n+k}, p_{n+k})$ of $\rho_B^{n+k}$ is isotropic, that is only depending on the norm of the vector $(q_1, p_1, \cdots, q_{n+k}, p_{n+k})$.
 The same also holds for the Q-function of the state, that is the probability distribution of the outcomes of the heterodyne measurements. 

Let us introduce another random variable $Z_n := 1/(2n) \sum_{i=1}^n q_{k+i}^2 + p_{k+i}^2$, corresponding to the norm of Bob's heterodyne measurements for the $n$ modes of  $\rho_B^{n+k}$ not measured during the test $\mathcal{T}$. 
We show in the appendix that the probability $\epsilon_\mathrm{test}$ of passing the test but $Z_n$ being much larger than $Y_\mathrm{test}$ is exponentially small in $k$ when the value of $Y_\mathrm{test}$ is chosen slightly larger the expected variance of Bob's measurement results (see Lemma \ref{step1bis}). 
In turn, this implies that the \emph{total} number of photons in the state $\rho_B^n$ is bounded by $O(n Y_\mathrm{test})$ (see Lemma \ref{step2bis}). Finally, we show that the projection over the space $\overline{\mathcal{H}}_B^{\otimes n}$ succeeds with high probability if $d_B = \mathrm{dim} \,\overline{\mathcal{H}}_B = O\left(\log\frac{2n}{\epsilon}\right)$ (see Lemma \ref{step3bis}). 
This finally provides a bound on $ || (1- \mathcal{P}) \circ \mathcal{T}||_\diamond$ and proves Theorem \ref{informal}.

We now put things together and establish that protocol $\mathcal{E}$ is secure against general attacks. First, choosing $d_A$ and $d_B$ on the order of $O(\log (n/\epsilon_\mathrm{test}))$, one obtains $ ||(1- \mathcal{P}) \circ \mathcal{T}||_\diamond \leq \epsilon_\mathrm{test}$.
Second, assuming that  the original protocol $\mathcal{E}_0$ is secure against collective attacks, the diamond norm $ ||\tilde{\mathcal{E}} - \tilde{\mathcal{F}}||_\diamond$ can be bounded by $2^{-c \delta^2 n + O\left(\log^2(n/\epsilon_\mathrm{test})\right)} $ using the postselection technique where the dimension of the relevant Hilbert space $\overline{\mathcal{H}}_A\otimes \overline{\mathcal{H}}_B$ is $d_A d_B = O\left(\log^2(n/\epsilon_\mathrm{test})\right)$ (see \cite{CKR09} for details). 
This shows that protocol $\mathcal{E}$ is $\epsilon$-secure against general attacks with 
\begin{equation}
\epsilon=2^{-c \delta^2 n + O\left(\log^2(n/\epsilon_\mathrm{test})\right)} + 2 \epsilon_\mathrm{test}.
\end{equation}

{\bf Conclusion}.---We have proved that Gaussian continuous-variable QKD protocols, using a Gaussian distribution of coherent states and homodyne or heterodyne measurements, are secure against arbitrary attacks. Our proof exploits the specific symmetries in phase-space of Gaussian QKD protocols and uses a simple test to ensure that the global state shared between Alice and Bob is well described by assigning a low dimensional Hilbert space to each mode.
This allows one to use the postselection technique introduced in Ref.\ \cite{CKR09} for discrete-variable protocols.
Our result greatly improves on a previous one using a de Finetti theorem which could not be applied to prove the security of protocols in realistic experimental implementations. 
Finally, our analysis indicates that in order to prove the security of any QKD protocol, one should exploit all the available symmetries of the protocol, beyond the traditional permutation.

{\bf Acknowledgements}.--- This work was supported by the SNF through the National Centre of Competence in Research ``Quantum Science and Technology'' (grant No.\  200020- 135048), the European Research Council (grant No.\ 258932), the Humboldt foundation and the F.R.S.-FNRS under project HIPERCOM.

\newpage

\begin{widetext}
\newpage

\appendix
\section*{Appendix}

In this appendix, we detail the various technical results used in the main text. 
In Appendix \ref{proof_hetero}, we explicitly state our main theorem for the continuous-variable protocol where Bob uses heterodyne detection. The proof of the main theorem uses three lemmas which are established in Appendices \ref{proof1}, \ref{proof2} and \ref{proof3}.
In Appendix \ref{proof_symmetry}, we detail why the protocol is invariant under the action of a network of beamsplitters and phase shifts, justifying the symmetry assumption made on Bob's quantum state. 

The rest of the appendix is devoted to protocols where Bob performs a homodyne detection instead of a heterodyne one. We state our main theorem in that case in Appendix \ref{proof_homo}. We prove in Appendix \ref{protocol_homodyne} that Bob's state can again be considered invariant under the action of beamsplitters and phase shifts. Our main theorem uses two of the same lemmas as in the heterodyne case and a variant of the third one which is established in Appendix \ref{proof2ter}.

\section{Main theorem for the heterodyne protocol}
\label{proof_hetero}

In order to make use of the relevant symmetries in phase space, the test itself should be invariant under the application of an arbitrary network of beamsplitters and phase-shifts on Bob's $(n+k)$ modes before he proceeds with his measurement. 
This can be enforced by actively symmetrizing the state, which can be done at the level of classical data (see Ref.\ \cite{lev12} for a discussion on this active symmetrization).
 
Bob randomly chooses a random unitary $U$ from the Haar measure on the unitary group $U(n+k)$.
Then, he symmetrizes his state thanks to the network of beamsplitters and phase shifts applying the transformation $U$ on the annihilation $(b_1, \cdots, b_{n+k})$ and creation operators $(b_1^\dagger, \cdots, b_{n+k}^\dagger)$ of his $(n+k)$ modes through 
\begin{equation}
b_i \rightarrow \sum_{j=1}^{n+k} U_{i,j} b_j \quad \text{and} \quad b_i^\dagger \rightarrow \sum_{j=1}^{n+k} U_{i,j}^* b_j^\dagger
\end{equation}
and finally measures his $(n+k)$ modes with a heterodyne detection. The state $\rho^{\otimes (k+n)}$ held by Bob after this symmetrization is called rotationally invariant. 

Crucially, Bob can also first measure his state with a heterodyne detection and only then implement $U$ by applying the symplectic transformation $S$ given by
\begin{equation}
S := \begin{pmatrix}
\mathrm{Re}(U) & - \mathrm{Im}(U) \\
\mathrm{Im}(U) & \mathrm{Re}(U)
\end{pmatrix}
\end{equation}
to his classical vector of measurements. This is true because the symmetrization in phase-space commutes with the heterodyne measurement (see Appendix \ref{proof_symmetry}  for details).

We denote $(q_1, p_1, \cdots, q_{n+k}, p_{n+k})$ the classical vector Bob obtains after this procedure. Thanks to the symmetrization, without loss of generality, the test $\mathcal{T}$ can be applied to the first $k$ modes, that is to the data $(q_1, p_1, \cdots, q_k, p_k)$. 

We now state our main theorem for the protocol with heterodyne detection. 
\begin{theo}[Heterodyne protocol]
\label{theo_hetero}
Let $\epsilon, Y_\mathrm{test} >0$ be fixed parameters. Let $Y_k = \frac{1}{k} \sum_{i=1}^k (q_i^2 + p_i^2)$ be the average of Bob's (squared) heterodyne measurement outcomes on the first $k$ modes of his state after symmetrization, and let $\rho^n$ be the state of the $n$ remaining modes. 
Let $d_B :=\frac{\log \left(4n/\epsilon\right)}{\log (1+1/d_0)}$ where $d_0 := g\left(\frac{\epsilon}{4}\right) Y_\mathrm{test}$ (and $g$ is defined in Eq.\ \ref{def_g}) and let $\overline{\mathcal{H}}_B = \mathrm{Span} \left\{ |0\rangle, \cdots, |d_B-1\rangle\right\}$ be the finite dimensional Hilbert space spanned by states with less than $d_B$ photons.
Then the probability that $Y_k \leq Y_\mathrm{test}$ and that the projection of $\rho^n$ onto $\overline{\mathcal{H}}_B^{\otimes n}$ fails is less than $\epsilon$. 
\end{theo}

In order to prove Theorem \ref{theo_hetero}, we need to introduce some operators acting on some subspace of $\mathcal{H}_B^{\otimes (n+k)}$. To keep notation simple, we use the subscript $k$ (resp. $n$) when the operator acts on $\mathcal{H}_B^{\otimes k}$ (resp. $\mathcal{H}_B^{\otimes n}$) corresponding to the first $k$ modes (resp. the last $n$ modes) of the symmetrized state. 
Let us define the POVM elements $T_k$, $T_n$, $U_n$ and $V_n$ on $\mathcal{H}^{\otimes n}$ as follows:
\begin{itemize}
\item $\mathcal{T}_k$ acting on $\mathcal{H}^{\otimes k}$ corresponding to a failed test, meaning that the value of the observable $Y_k$ is larger than $T_\mathrm{test}$:
\begin{equation}
\mathcal{T}_k := \frac{1}{\pi^k} \int_{\sum_{i=1}^k |\alpha_i|^2 \geq Y_\mathrm{test}} |\alpha_1\rangle\!\langle\alpha_1| \cdots |\alpha_k\rangle\!\langle \alpha_k| \mathrm{d}\alpha_1 \cdots \mathrm{d} \alpha_k, 
\end{equation}
\item $T_n$ is the projector onto products of coherent states $|\alpha_1\rangle \cdots |\alpha_n\rangle \in \mathcal{H}_B^{\otimes n}$ such that $\sum_{i=1}^n |\alpha_i|^2 \geq n d_0$:
\begin{equation}
T_n := \frac{1}{\pi^n} \int_{\sum_{i=1}^n |\alpha_i|^2 \geq n d_0} |\alpha_1\rangle\!\langle\alpha_1| \cdots |\alpha_n\rangle\!\langle \alpha_n| \mathrm{d}\alpha_1 \cdots \mathrm{d} \alpha_n,
\end{equation}
\item $U_n$ is the projector onto the subspace of $\mathcal{H}^{\otimes n}$ spanned by states with more than $n d_0$ photons
\begin{equation}
U_n := \sum_{m= n d_0+1}^\infty \Pi_m^n,
\end{equation}
where we introduced the projector $\Pi_m^n$ on the subspace spanned by $n$-mode states containing $m$ photons:
\begin{equation}
\Pi_m^n =\sum_{m_1+ \cdots+ m_n=m} |m_1 \cdots m_n\rangle \!\langle m_1\cdots m_n|.
\end{equation}
\item $V_n$ is the projector onto the subspace of $\mathcal{H}^{\otimes n}$ such that at least one mode contains at least $d_B$ photons:
\begin{equation}
V_n :=  \mathbbm{1}-P_B^{\otimes n},
\end{equation}
where $P_B := |0\rangle \!\langle 0| + \cdots + |d_B-1\rangle\! \langle d_B-1|$.
\end{itemize}
With these notations, Theorem \ref{theo_hetero} simply gives an upper bound on $p_\mathrm{bad} :=   \mathrm{tr} \left[ V_n (1-\mathcal{T}_k)  \rho^{n+k} \right]$ for any rotationally invariant state $\rho^{n+k} \in \mathcal{H}_{B}^{\otimes (n+k)}$. The quantity  $p_\mathrm{bad}$ is the probability that the state passes the test and that the projection on the finite-dimensional subspace $\overline{\mathcal{H}}_B^{\otimes n}$ fails. 

The proof of the theorem uses the following variants of three technical lemmas proven in Sections \ref{proof1}, \ref{proof2} and \ref{proof3} of this appendix. 
We first state Lemma \ref{step1bis} which is a corollary of a result proven in Section \ref{proof1}.
\begin{lemma}
\label{step1bis}
Let $\mathbf{X} = (X_1, \cdots, X_{k+n})$ be a vector of  $\mathbbm{C}^{n+k}$. Let $U$ be a random unitary transformation of $U(n+k)$ drawn from the Haar measure, and define $\mathbf{Y} = U \mathbf{X}$. Then
\begin{equation}
\mathrm{Pr}\left[\frac{1}{n} \sum_{i=1}^n |Y_{k+i}|^2  \geq g(\delta)\frac{1}{k} \sum_{i=1}^k |Y_k|^2 \right] \leq \delta
\end{equation}
where 
\begin{equation}
g(\delta) = \frac{  1 + 2\sqrt{\frac{ \log (1/\delta)}{2n}} +\frac{2  \log (2/\delta)}{2n} }{1 - \sqrt{\frac{2}{k} \log \left(\frac{2}{\delta} \right)}} \label{def_g}.
\end{equation}
\end{lemma}
By construction, the vector $Y = (Y_1, \cdots, Y_{n+k})$ is uniformly distributed on the complex sphere in $\mathbbm{C}^{n+k}$ with radius $||Y||$, and consequently, the real vector $(\mathrm{Re}(Y_1), \mathrm{Im}(Y_1), \cdots, \mathrm{Re}(Y_{n+k}), \mathrm{Im}(Y t_{n+k}))$ is uniformly distributed on the corresponding real sphere in $\mathbbm{R}^{2(n+k)}$. Lemma \ref{step1bis} is then a special case of Lemma \ref{step1}.

The following lemma is proved in Section \ref{proof2}. 
\begin{lemma}
\label{step2bis}
\begin{equation}
U_n \leq 2 T_n.
\end{equation}
\end{lemma}

Our final lemma quantifies the maximum number of photons in a single mode for a rotationally invariant state with $n d$ photons in $n$ modes, except with a small probability:
\begin{lemma}
\label{step3bis}
Let $m_1, \cdots, m_n$ be the random variables corresponding to photon counting measurements of the $n$ modes of the state $\sigma_{n d}^n$ which is a uniform mixture of states with $n d$ photons in $n$ modes. Then, the following bound holds:
\begin{equation}
\mathrm{Pr}\left[ \max_{i=1\cdots n} m_i \geq  \frac{\log \left(\frac{2n}{\epsilon}\right)}{\log (1+1/d)} \right] \leq \epsilon.
\end{equation}
\end{lemma}

\begin{proof}[Proof of Theorem \ref{theo_hetero}]
We fix $d_0 := g\left(\frac{\epsilon}{4}\right) Y_\mathrm{test}$ and $d_B :=\frac{\log \left(4n/\epsilon\right)}{\log (1+1/d_0)}$.
From Lemma \ref{step1bis}, we know that 
\begin{equation}
\mathrm{tr} \, T_n (\mathbbm{1}_k-T_k) \rho^{n+k} \leq \frac{\epsilon}{4}
\end{equation}
for any rotationally invariant state $\rho^{n+k}$. 
From Lemma \ref{step2bis}, we obtain
\begin{equation}
\mathrm{tr} \, U_n (\mathbbm{1}_k-T_k) \rho^{n+k} \leq \frac{\epsilon}{2}
\end{equation}
and Lemma \ref{step3bis} shows that
\begin{equation}
\mathrm{tr} \,  (1-U_n) V_n \rho^{n+k} \leq \frac{\epsilon}{2}
\end{equation}
for any rotationally invariant state $\rho^{n+k}$. 

Using that $V_n \leq V_n(1-U_n) + U_n$, one finally has:
\begin{eqnarray}
p_\mathrm{bad, Bob} &:=&   \mathrm{tr} \left[ (V_n \circ \mathcal{T}) \rho^{n+k} \right]\\
&=&  \mathrm{tr} \left[ V_n (\mathbbm{1}_k - T_k)\rho^{n+k} \right]\\
&\leq& \mathrm{tr} \left[(1-U_n)  V_n (\mathbbm{1}_k - T_k)\rho^{n+k} \right] +  \left[ U_n (\mathbbm{1}_k - T_k)\rho^{n+k} \right]\\
&\leq&\mathrm{tr} \left[(1-U_n)  V_n\rho^{n+k} \right] + 2 \mathrm{tr} \, T_n (\mathbbm{1}_k-T_k) \rho^{n+k}  \\
&\leq& \frac{\epsilon}{2} + 2  \times \frac{\epsilon}{4}   \\
&\leq& \epsilon.
\end{eqnarray}

\end{proof}


\section{Concentration of measure on the sphere}
\label{proof1}
In this section, we establish the following result which implies Lemma \ref{step1bis}. 
\begin{lemma}
\label{step1}
If the vector $\mathbf{X} = (X_1, \cdots, X_{k+n})$ is uniformly distributed on the unit sphere of $\mathbbm{R}^{n+k}$, then 
\begin{equation}
\mathrm{Pr}\left[\frac{1}{n} \sum_{i=1}^n X_{k+i}^2  \geq g(\delta)\frac{1}{k} \sum_{i=1}^k X_k^2 \right] \leq \delta
\end{equation}
where 
\begin{equation}
g(\delta) = \frac{  1 + 2\sqrt{\frac{ \log (2/\delta)}{n}} +\frac{2  \log (2/\delta)}{n} }{1 - 2\sqrt{\frac{1}{k} \log \left(\frac{2}{\delta} \right)}}.
\end{equation}
\end{lemma}

We do not prove Lemma \ref{step1} directly because manipulating normalized vectors on the sphere is not very convenient. We use instead the natural invariance of the problem and first show that it is sufficient to prove the lemma for independent normal variables instead of vectors on the unit sphere. 
This is the case because a uniformly chosen vector on the sphere can be obtained by drawing $n$ independent normal variables and normalizing the corresponding vector. 

Let $X_1, \cdots, X_n$ be such independent normal random variables: $X_i \sim  \mathcal{N}(0,1)$, and let us define the following quantities:
\begin{equation}
Y_k = \frac{1}{k} \sum_{i=1}^k X_i^2 \quad \text{and} \quad Z_n = \frac{1}{n} \sum_{i=1}^n X_{k+i}^2.
\end{equation}
Note that the normalized random vector $\mathbf{\tilde{X}} = \frac{1}{\sqrt{\sum_{i=1}^{n+k} X_i^2}} (X_1, \cdots, X_{n+k})$ is uniformly distributed on the unit sphere of $\mathbbm{R}^{n+k}$. 
In particular, it is sufficient to prove that 
\begin{equation}
\label{equivalent}
\mathrm{Pr}\left[Z_n  \geq g(\delta) Y_k \right] \leq \delta
\end{equation} in order to establish the lemma. 

Let us proceed with the proof of Eq.\ \ref{equivalent}.
We first notice that for any $A >0$, 
\begin{equation}
\mathrm{Pr}\left[Z_n  \geq g(\delta) Y_k \right]  \leq  \mathrm{Pr}\left[Y_k  \leq A \right]  + \mathrm{Pr}\left[Z_n  \geq g(\delta) A \right] .
\end{equation}
We can now bound this two probabilities using the fact that $Y_k$ and $Z_n$ are independent random variables, with a $\chi^2(k)$ and a $\chi^2(n)$ distribution, respectively. 
To this end, we use two bounds on $\chi^2$ distributions established by Laurent and Massart \cite{LM00}:
\begin{equation}
\mathrm{Pr}\left[Y_k \leq 1 - 2\sqrt{\frac{x}{k}}\right] \leq \exp (-x)
\quad \text{and} \quad
\mathrm{Pr}\left[Z_n \geq 1 + 2\sqrt{\frac{x}{n}} + \frac{2x}{n}\right] \leq \exp (-x).  
\end{equation}
Choosing $x = \log(2/\delta)$ in both cases gives:
\begin{equation}
\mathrm{Pr}\left[Y_k \leq 1 - 2\sqrt{\frac{ \log(2/\delta)}{k}}\right] \leq \frac{\delta}{2}
\quad \text{and} \quad
\mathrm{Pr}\left[Z_n \geq 1 + 2\sqrt{\frac{ \log(2/\delta)}{n}} + \frac{2 \log(2/\delta)}{n}\right] \leq \frac{\delta}{2}.  
\end{equation}
Taking $A :=  1 - 2\sqrt{\frac{1}{k} \log \left(\frac{2}{\epsilon} \right)}$ concludes the proof of Lemma \ref{step1}.


\section{Proof of Lemma \ref{step2bis}}
\label{proof2}
In this section, we prove Lemma \ref{step2bis} which we recall here.
\begin{replemma}{step2bis}
\label{step2restated}
Let $T_n$ and $U_n$ be defined as
\begin{equation}
T_n := \frac{1}{\pi^n} \int_{\sum_{i=1}^n |\alpha_i|^2 \geq n d_0} |\alpha_1\rangle\!\langle\alpha_1| \cdots |\alpha_n\rangle\!\langle \alpha_n| \mathrm{d}\alpha_1 \cdots \mathrm{d} \alpha_n,
\end{equation}
and
\begin{equation}
U_n := \sum_{m= n d_0+1}^\infty \Pi_m^n \quad \text{with} \quad \Pi_m^n =\sum_{m_1+ \cdots+ m_n=m} |m_1 \cdots m_n\rangle \!\langle m_1\cdots m_n|.
\end{equation}
Then, the following inequality holds:
\begin{equation}
U_n \leq 2 T_n.
\end{equation}
\end{replemma}

\subsection{Some preliminaries}

The following integrals will be useful. For $a>0$, let us define:
\begin{equation}
I_n(a) = \int_{y_i \geq 0, \sum_{i=1}^n y_i \geq a} e^{-y_1 - y_2 \cdots -y_n} \mathrm{d}y_1 \cdots \mathrm{d}y_n
\end{equation}
and 
\begin{equation}
J_n(k,a) = \int_{y_i \geq 0, \sum_{i=1}^n y_i \geq  a}\frac{y_1^k}{k!} e^{-y_1 - y_2 \cdots -y_n} \mathrm{d}y_1 \cdots \mathrm{d}y_n.
\end{equation}
These integrals can be computed explicitly.
\begin{lemma}
\label{prel1}
\begin{align}
I_n(a) &= e^{-a} \sum_{k = 0}^{n-1} \frac{a^k}{k!}\\
J_n(k,a) & =\frac{\Gamma(k+1,a)}{\Gamma(k+1,0)} + e^{-a}   \sum_{m = k+1}^{k+n} \frac{a^{m}}{m!}.
\end{align}
\end{lemma}

\begin{proof}
The first equality is proved by induction. It is clear that $I_1(a) = e^{-a}$.
Then:
\begin{align}
1-I_{n+1}(a) &= \int_{y_i \geq 0, \sum_{i=1}^{n+1} y_i \leq a} e^{-y_1 - y_2 \cdots -y_{n+1}} \mathrm{d}y_1 \cdots \mathrm{d}y_{n+1}\\
&=  \int_0^{a} \mathrm{d}y_{n+1}  e^{-y_{n+1}} \int_{y_i \geq 0, \sum_{i=1}^n y_i \leq d_0-y_{n+1}}  e^{-y_1 - y_2 \cdots -y_n} \mathrm{d}y_1 \cdots \mathrm{d}y_n\\
&= \int_0^{a} \mathrm{d}y_{n+1}  e^{-y_{n+1}} (1-I_n(a-y_{n+1}))\\
&=  \int_0^{a} \mathrm{d}y_{n+1}  e^{-y_{n+1}} \left(1-e^{-a+y_{n+1}} \sum_{k =0}^{n-1} \frac{(a-y_{n+1})^k}{k!}\right)\\
&= 1- e^{-a} - e^{-a}  \int_0^{a} \sum_{k = 0}^{n-1} \frac{(a-y)^k}{k!}  \mathrm{d}y\\
&= 1- e^{-a} - e^{-a}  \sum_{k = 0}^{n-1} \frac{a^{k+1}}{(k+1)!}\\
&= 1 - e^{-a} \sum_{k=0}^{n} \frac{a^k}{k!}\\
\end{align}

\begin{align}
J_n(k,a) &=  1- \int_{y_i \geq 0, \sum_{i=1}^n y_i \leq  a}\frac{y_1^k}{k!} e^{-y_1 - y_2 \cdots -y_n} \mathrm{d}y_1 \cdots \mathrm{d}y_n\\
&= 1- \int_0^{a} \mathrm{d}y_1 \frac{y_1^k}{k!} e^{-y_1} \int_{y_i \geq 0, \sum_{i=2}^n y_i \leq a  - y_1 }  e^{-y_2 - y_3 \cdots -y_n} \mathrm{d}y_2 \cdots \mathrm{d}y_n\\
&=  1- \int_0^{a} \mathrm{d}y_1 \frac{y_1^k}{k!} e^{-y_1}  (1-I_{n-1}(a - y_1))\\
&=  1- \int_0^{a} \mathrm{d}y_1 \frac{y_1^k}{k!}  e^{-y_1}  \left( 1-e^{-a+y_1} \sum_{m = 0}^{n-1} \frac{(a-y_1)^m}{m!}\right)\\
&=  \frac{\Gamma(k+1,a)}{\Gamma(k+1,0)} + e^{-a}  \sum_{m = 0}^{n-1}  \int_0^{a}   \frac{y^k (a-y)^m}{k! \,m!} \mathrm{d}y
\end{align}
where $\Gamma(s,x)=\int_x^\infty t^{s-1}e^{-t}\mathrm{d}t$ is the incomplete gamma function
Using the fact that
\begin{equation}
\int_0^a x^k (a-x)^m \mathrm{d}x = \frac{k! m! a^{k+m+1}}{(k+m+1)!},
\end{equation}
one obtains
\begin{align}
J_n(k,a) &=  \frac{\Gamma(k+1,a)}{\Gamma(k+1,0)} + e^{-a}   \sum_{m = 0}^{n-1} \frac{(a)^{k+m+1}}{(k+m+1)!}\\
&= \frac{\Gamma(k+1,a)}{\Gamma(k+1,0)} + e^{-a}   \sum_{m = k+1}^{k+n} \frac{a^{m}}{m!}\\
\end{align}
\end{proof}

\subsection{Proof of Lemma \ref{step2bis}}

Integrating over the $n$ phases gives:
\begin{equation}
T_n = \sum_{k_1, \cdots, k_n} \int_{x_i\geq 0, \sum_{i=1}^n x_i \geq n d_0} \prod_{i=1}^n e^{-x_i} \frac{x_i^{k_i}}{k_i!} \mathrm{d}x_i |k_1 \cdots k_n\rangle \!\langle k_1\cdots k_n|.
\end{equation}
Because of its rotation invariance in phase-space, the operator $T_n$ can be written as a mixture of $\Pi_k^n$. Let us note $q_k \geq 0$ the corresponding coefficients: $T_n = \sum_{k=0}^\infty q_k  \Pi_k^n$. Considering the term $\langle k,0,\cdots, 0 |T_n |k, 0,\cdots, 0\rangle$, it is easy to see that $q_k = J_n(k,n d_0)$. 

The proof is then immediate by noticing that the sequence $\frac{\Gamma(k+1,a)}{\Gamma(k+1,0)} = q_k -  e^{-a}   \sum_{m = k+1}^{k+n} \frac{a^{m}}{m!}$ (where we used the result of Lemma \ref{prel1}) is positive and increasing with $k$ for all $a \geq0$. Here, $\Gamma(s,x) := \int_x^\infty t^{s-1} e^{-t} \mathrm{d}t$ refers to the incomplete Gamma function. 
This means that for $k \geq n d_0+1$, 
\begin{equation}
q_k \geq  \frac{\Gamma(n d_0+1,n d_0)} {\Gamma(n d_0+1,0)} +  e^{-a}   \sum_{m = k+1}^{k+n} \frac{a^{m}}{m!} \geq \frac{\Gamma(n d_0+1,n d_0)} {\Gamma(n d_0+1,0)}  \geq \frac{1}{2},
\end{equation}
where we used that $\frac{\Gamma(x+1,x)}{\Gamma(x+1,0)}$ is lower bounded by $1/2$ for all $x\geq0$. 
This allows us to conclude that
\begin{equation}
U_n  \leq 2 T_n. 
\end{equation}


\section{Proof of Lemma \ref{step3bis}}
\label{proof3}

The set of vectors $\mathbf{X} = (X_1, \cdots, X_n)$ such that $\sum_{i=1}^n X_i = k$ and $X_1 \geq m$ contains $a_{k-m}^n := {n+k-m-1 \choose k-m}$ elements. 
Let us note $p_k(m,n)$ the probability that the maximum of $X_i$ is greater than $m$ if one measures the photon number for the state $\sigma_k^n$:
\begin{equation}
p_k(m,n) = \mathrm{Pr}\left[ \max_{i=1\cdots n} X_i \geq m \; \text{s.t.} \;\sum_{i=1}^n X_i  = k\right].
\end{equation}
The union bound gives:
\begin{eqnarray}
p_k(m,n)  &\leq & \mathrm{Pr}\left[X_1 \geq m  \; \text{s.t.} \; \sum_{i=1}^n X_i  = k\right] + \cdots + \mathrm{Pr}\left[X_n \geq m  \; \text{s.t.} \; \sum_{i=1}^n X_i  = k\right] \\
& \leq & n \frac{a_{k-m}^n}{a_k^n} = n \frac{(n+k-m-1)! k!}{(n+k-1)! (k-m)!} = \frac{n (n+k)}{n+k-m} \frac{(n+k-m)! k!}{(n+k)! (k-m)!}.
\end{eqnarray}

Let $x>0$, then Stirling approximation formula reads (here $\log$ is the natural logarithm):
\begin{equation}
n x \log n + n x (\log x -1) + \frac{1}{2} \log n + \frac{1}{2} \log x + \log \sqrt{2\pi} \leq \log (nx)! \leq n x \log n + n x (\log x -1) + \frac{1}{2} \log n + \frac{1}{2} \log x + 1.
\end{equation}

Let us introduce the variables $d$ and $\delta$ such that $k  = d n$ and $m = \delta n$. 
Then
\begin{align}
\log  \frac{(n+k-m)! k!}{(n+k)! (k-m)!} & = \log (n(d+1-\delta)) ! + \log (nd)! - \log (n(d+1))! - \log (n(d-\delta))!\\
& \leq- n \left\{g(d) - g(d-\delta) \right\} + \frac{1}{2} \log \frac{d(d+1-\delta)}{(d+1)(d-\delta)} + 2 - \log 2\pi
\end{align}
where 
\begin{equation}
g(x)  = (x+1) \log (x+1) -x \log x. 
\end{equation}
This gives
\begin{equation}
\log p_k(m,n) \leq - n \left\{g(d) - g(d-\delta) \right\}+ \log n + \frac{1}{2} \log \frac{d(d+1)}{(d-\delta)(d-\delta+1)} + 2 - \log 2\pi
\end{equation}

The function $g$ is concave which implies that $g(d) - g(d-\delta) \geq \delta g'(d) = \delta \log (1+1/d)$, and therefore
\begin{align}
- \log p_k(m,n) &\geq n\delta \log (1+1/d) -\log n+ \frac{1}{2} \log \frac{(d+1)(d+1-\delta)} {d(d-\delta)} - 2 + \log 2\pi\\
&\geq n\delta \log (1+1/d) -\log n - 2 + \log 2\pi\\
&\geq n\delta \log (1+1/d) - \log n - \log 2\\
\end{align}

Choosing $m  = \frac{\log(2n/\epsilon)}{\log (1+1/d)}$ gives $p_k(m,n) \leq \epsilon$ and proves Lemma \ref{step3bis}.


\section{Symmetry of the state for heterodyne detection}
\label{proof_symmetry} 
In this section, we show that the symplectic transformation applied in phase-space commutes with the heterodyne detection. 
The compact subgroup of the symplectic group $\mathrm{Sp}(2N,\mathbbm{R})$ consisting of phase shifts and beamsplitters is usually noted $K(n)$ in the literature and is isomorphic to the unitary group $U(N)$ (see for instance Ref.\ \cite{ADMS95}). We note $\vec{a}:=(\hat{a}_1, \cdots, \hat{a}_N)$ and $\vec{a}^\dagger := (\hat{a}_1^\dagger, \cdots, \hat{a}_N^\dagger)$ the vectors of annihilation and creation operators of the $N$ modes considered. 
Then, in the Heisenberg picture, under a symplectic transformation, the $\hat{a}$'s and $\hat{a}^\dagger$'s transform independently as:
\begin{equation}
\vec{a} \rightarrow U \vec{a}, \quad \text{and} \quad \vec{a}^\dagger \rightarrow U^* \vec{a}^\dagger,
\end{equation}
where $U$ is a unitary matrix.

Moreover, defining $V = \mathrm{Re}(U)$ and $W=-\mathrm{Im}(U)$ the real and imaginary parts of $U$ such that $U = V - iW$, the displacement vector $(\vec{x},\vec{p})^T:=(x_1, \cdots, x_N, p_1, \cdots, p_N)^T$ is transformed as
\begin{equation}
\begin{pmatrix}
\vec{x} \\ \vec{p} \\
\end{pmatrix}
\rightarrow
\begin{pmatrix}
V & W \\
-W & V\\
\end{pmatrix}
\begin{pmatrix}
\vec{x} \\ \vec{p} \\
\end{pmatrix}.
\end{equation}

In the quantum key distribution protocol, both Alice and Bob perform a heterodyne measurement of their respective $n+k$ modes. The probability distribution of their outcomes is given by the $Q$-function of the state $\rho_{AB}^{n+k} \in (\mathcal{H}_A\otimes \mathcal{H}_B)^{\otimes (n+k)}$: $Q_\rho(\vec{x}_A, \vec{p}_A, \vec{x}_B, \vec{p}_B)$. 
From the description given above of the subgroup $K(n)$, it appears that the $Q$-function associated with the "rotated" state $\rho':=(\mathcal{U}_A \otimes \mathcal{U}_B^*) \rho_{AB}^{n+k} (\mathcal{U}_A\otimes \mathcal{U}_B^*)^\dagger$ (where $\mathcal{U}$ is the representation of the symplectic transformation corresponding  to the unitary $U$) is simply:
\begin{equation}
Q_{\rho'}(\vec{x}_A', \vec{p}_A', \vec{x}_B', \vec{p}_B') = Q_\rho(V \vec{x}_A - W \vec{p}_A, W \vec{x}_A + V \vec{p}_A, V\vec{x}_B+W \vec{p}_B, -W \vec{x}_B + \vec{V}p_B).
\end{equation}
This local (and classical) transformation of the coordinates can be applied by Alice and Bob. In other words, the quantum transformation corresponding to the networks of beamsplitters and phase-shifts commutes with the heterodyne measurement. 

Moreover, if one makes sure that the test (i.e., the choice of the random $2k$-dimensional subspace of $\mathbbm{R}^{2(n+k)}$) respects the symmetry above, then the post processing of the QKD protocol commutes with the map $\mathcal{U}_A\otimes \mathcal{U}_B^*$, meaning that the state $\rho_{AB}^{n+k}$ can be considered invariant under such maps. In particular, the state held by Bob, $\rho_B^{n+k} := \mathrm{tr}_A \rho_{AB}^{n+k}$ is invariant under any $\mathcal{U}$ consisting in beamsplitters and phase shifts.


\newpage

\section{Main theorem for the homodyne protocol}
\label{proof_homo}

We now consider a protocol where Alice sends coherent states with a Gaussian modulation and Bob performs homodyne detection with a random quadrature chosen uniformly in $[0, 2\pi]$ for each of his $n+k$ modes.

In the following, we note $X_1, \cdots, X_{n+k}$ the random variables corresponding to the $n+k$ quadrature measurement outcomes for Bob. We assume that for each mode, Bob chooses a random direction $\theta \in[0, \pi/2]$ in phase space. Then, he chooses to measure the quadrature either along $\theta$ (which we call a $q$ measurement) or along $\theta+ \pi/2$ (this is the $p$ measurement). We will show in Section \ref{protocol_homodyne} that Bob's state can be assumed to rotationally invariant. 

Our main result is summarized by the following theorem. 
\begin{theo}
\label{theo_homo}
Let $\epsilon, Y_\mathrm{test} >0$ be fixed parameters. Let $Y_k = \frac{1}{k} \sum_{i=1}^k X_i^2$ be the average of Bob's (squared) homodyne measurement outcomes on the first $k$ modes of his symmetrized state, and let $\rho^n$ be the state of his $n$ remaining modes. Let $d_0:= 2 g\left(\frac{\epsilon}{16}\right) Y_k$. We choose $n$ large enough so that $e^{-\beta n} \leq \frac{\epsilon}{16}$  with  $\beta := c_0 d_0 - \frac{\log d_0}{2}$ and $c_0 := (1-1/\sqrt{2})^2$. Let $d_B=  \frac{\log \left(4n/\epsilon\right)}{\log (1+1/d_0)}$.  Let $\overline{\mathcal{H}}_B = \mathrm{Span} \left\{ |0\rangle, \cdots, |d_B-1\rangle\right\}$ be the finite dimensional Hilbert space spanned by states with less than $d_B$ photons.
Then the probability that $Y_k \leq Y_\mathrm{test}$ and that the projection of $\rho^n$ onto $\overline{\mathcal{H}}_B^{\otimes n}$ fails is less than $\epsilon$. \end{theo}

The main novelty compared to the heterodyne protocol is the introduction of the operator $W_n$, corresponding to the projection on the event $\left[Z_n \geq d_0/2\right]$, where $Z_n = \frac{1}{n} \sum_{i=1}^k X_{k+i}^2$. Let us define the random variable $s_i \in \{q_i, p_i \}$ corresponding to a quadrature measurement for the mode $i$, either $q$ or $p$. Then, we can define a string $\mathbf{s} \in \{q, p \}^n$ as being the $n$ quadrature measurement outcomes when measuring the state $\rho^n$. With these notations, 
\begin{equation}
W_n= \frac{1}{2^n} \sum_{\mathbf{s} \in \{x, p \}^n} P^{Z_n(\mathbf{s}) \geq d_0/2}.
\end{equation}
We prove the following result in Section \ref{proof2ter}.
\begin{lemma}
\label{step2ter}
\begin{equation}
T_n \leq 2 W_n + e^{-\beta n} \mathbbm{1}_n
\end{equation}
where $\beta = c_0 d_0 - \frac{\log d_0}{2}$ and $c_0 = (1-1/\sqrt{2})^2$.
\end{lemma}
The proof of Theorem \ref{theo_homo} then follows exactly the same lines at that of Theorem \ref{theo_hetero}.



\section{Symmetry of Bob's state for the protocol with homodyne detection}
\label{protocol_homodyne}

Let $\rho^{n+k} = \sum_{\vec{i},\vec{j}} a_{\vec{i},\vec{j}} \, |i_1, i_2,\cdots, i_{n+k}\rangle\!\langle j_1,j_2,\cdots, j_{n+k}|$ be Bob $n+k$-mode state. 

The measurement protocol is the following:
\begin{itemize}
\item for each mode $j$, Bob draws $\theta_j$ uniformly in $[0,2\pi]$ and measures the quadrature along $\cos \theta_j q_j + \sin \theta_j p_j$. He obtains a measurement outcome $x_j$. It is actually sufficient to pick $\theta$  from the set $\{0, 2\pi \frac{1}{d_B}, 2\pi \frac{2}{d_B},\cdots, 2\pi \frac{d_B-1}{d_B}\}$.  
\item Bob then randomly chooses an orthogonal transformation $R$ in $\mathbbm{R}^{n+k}$ and applies it to his vector $\vec{x} = (x_1, \cdots, x_{n+k})$. 
\item finally, Bob informs Alice of his choices of $\theta_j$ and $R$. 
\end{itemize}
Crucially, the transformation $R$ can be equivalently obtained by applying a network of beamsplitters on the $n+k$ modes. 

Because of the random choice of the measured quadratures, the state $\rho^{n+k}$ can be considered invariant under the application of $U(\theta_j) = e^{i \theta_j a^\dagger_j a_j}$. This means that
\begin{align}
\rho^{n+k} & \propto \int_{\vec{\theta} \in [0,2\pi]^n} \left (\prod_{j=1}^n U(\theta_j)\right) \rho \left (\prod_{j=1}^n U(-\theta_j)\right)\mathrm{d} \vec{\theta}\\
& \propto \sum_{\vec{i},\vec{j}} a_{\vec{i},\vec{j}} |i_1, i_2,\cdots, i_{n+k}\rangle\!\langle j_1,j_2,\cdots, j_{n+k}|  \int_{\vec{\theta} \in [0,2\pi]^{n+k}} e^{i \theta_1 (i_1-j_1) }\cdots e^{i \theta_1 (i_{n+k}-j_{n+k}) } \mathrm{d} \vec{\theta}\\
& \propto \sum_{\vec{i}} a_{\vec{i},\vec{i}} |i_1, i_2,\cdots, i_{n+k}\rangle\!\langle i_1,i_2,\cdots, i_{n+k}|  .
\end{align}

Because of the random rotation of the measurement results, the state can also be considered invariant under the action of any network of beamsplitters. In particular, it should be invariant when swapping any two modes and when applying infinitesimal beamsplitters. 
The first condition shows that the coefficient $a_{\vec{i} } := a_{\vec{i}, \vec{i}}$ is invariant when permuting the coordinates of the vector $\vec{i}$. The invariance under infinitesimal beamsplitters guarantees that $a_{i_1, i_2, i_3, \cdots, i_{n+k}} = a_{i_1+1, i_2-1, i_3, \cdots, i_{n+k}} = a_{i_1+\cdots+i_{n+k},0, \cdots, 0}$.  In particular, the coefficient $a_{\vec{i}}$ only depends on the total number of photons $i$ in the $n$ modes.

Finally, the state $\rho^{n+k}$ can be assumed to be a mixture of the states $\sigma_i^{n+k}$ defined as:
\begin{equation}
\sigma_i^{n+k} = \frac{1}{{n+k+i-1 \choose i} } \sum_{\sum_{i_j} = i} |i_1, \cdots, i_{n+k}\rangle \!\langle i_1, \cdots, i_{n+k}|.
\end{equation}


\section{Proof of Lemma \ref{step2ter}}
\label{proof2ter}

\subsection{Preliminaries}

Let us define $F(\vec{a})$ as
\begin{equation}
F(\vec{a}) = \frac{1}{\pi^{n/2}} \int_{\|\vec{z}\|^2 \geq n d_0} \mathrm{d} \vec{z} e^{-\|\vec{z}-\vec{a}\|^2}. 
\end{equation}
The spherical symmetry of the function guarantees that $F(\vec{a})$ only depends on the norm of $\vec{a}$. Let us note $a = \|\vec{a}\|$. In the following, it is sometimes useful to think of the vector $\vec{a}$ as $\vec{a}=(a,0,\cdots, 0)$. The following bound will be useful in the proof of Lemma \ref{step2ter}.
\begin{lemma}
\label{prel2}
For any $d_0>0$,
\begin{equation}
F\left( \sqrt{\frac{n d_0}{2}}\right) \leq e^{- \beta n}, \quad \mathrm{with} \quad\beta = \left(1-\frac{1}{\sqrt{2}}\right)^{\!\!\! 2} \! d_0 - \frac{\log d_0}{2}.
\end{equation}
\end{lemma}

\begin{proof}
Let us first compute the following $n$-dimensional integral (we use spherical coordinates and recall that the surface of the unit sphere in $\mathbb{R}^n$ is $\frac{2\pi^{n/2}}{\Gamma(n/2)}$):
\begin{align}
\frac{1}{\pi^{n/2}} \int_{\sum_{i=1}^n z_i^2 \geq b^2} \exp \left(-\sum_{i=1}^n z_i^2 \right)  \mathrm{d}z_1 \cdots  \mathrm{d}z_n  =& \frac{1}{\Gamma(n/2)} \int_{R =b^2}^\infty R^{n/2-1} e^{-R} \mathrm{d}R \nonumber\\
=& \frac{\Gamma\left(\frac{n}{2},b^2\right)}{\Gamma\left(\frac{n}{2},0\right)}. \label{useful_eq}
\end{align}
Then, translating the variable $\vec{z}$ by $\vec{a}$, one obtains for $a \leq \sqrt{n d_0}$,
\begin{eqnarray}
F(a) &\leq & \frac{1}{\pi^{n/2}} \int_{\|\vec{z}\|^2 \geq (\sqrt{n d_0}-a)^2} \mathrm{d} \vec{z} e^{-\|\vec{z}\|^2}\\
&\leq& \frac{\Gamma\left(\frac{n}{2},(\sqrt{n d_0}-a)^2\right)}{\Gamma\left(\frac{n}{2},0\right)},
\end{eqnarray} 
where the first inequality holds because the integration domain contains the one of the definition of $F(a)$, and the second is the application of Eq.\ \ref{useful_eq} with $b = \sqrt{n d_0-1}$. 
Choosing $a = \sqrt{\frac{n d_0}{2}}$ finally gives
\begin{equation}
F\left( \sqrt{\frac{n d_0}{2}}\right) \leq \frac{\Gamma\left(\frac{n}{2},n d_0 c_0\right)}{\Gamma\left(\frac{n}{2},0\right)},
\end{equation}
with
\begin{equation}
c_0 = \left(1-\frac{1}{\sqrt{2}}\right)^2. 
\end{equation}

Let $X$ be a random variable with a Poisson distribution of parameter $\lambda = n c_0 d_0$ and assume $n$ to be even, then 
\begin{equation}
\mathrm{Pr}\left[X\leq n/2\right] =  \frac{\Gamma\left(\frac{n}{2},n d_0 c_0\right)}{\Gamma\left(\frac{n}{2},0\right)},
\end{equation}
which implies that
\begin{equation}
F\left( \sqrt{\frac{n d_0}{2}}\right) \leq \mathrm{Pr}\left[X\leq n/2\right].
\end{equation}

Chernoff bound applied to a Poisson distribution of parameter $\lambda$ gives:
\begin{equation}
\mathrm{Pr}[X \leq (1-\delta)\lambda]  \leq \left(\frac{e^{-\delta}}{(1-\delta)^{1-\delta} }\right)^\lambda,
\end{equation}
which gives here
\begin{equation}
F\left( \sqrt{\frac{n d_0}{2}}\right) \leq e^{- \tilde{\beta} n},
\end{equation}
with 
\begin{equation}
\tilde{\beta} = c_0 d_0 - \frac{1+\log(2 c_0 d_0)}{2}.
\end{equation}
Using the fact that $1 + \log(2 c_0) \leq 0$, one obtains
\begin{equation}
F\left( \sqrt{\frac{n d_0}{2}}\right) \leq e^{- \beta n},
\end{equation}
with
\begin{equation}
\beta = c_0 d_0 - \frac{\log d_0}{2}.
\end{equation}
\end{proof}

\subsection{Proof of the lemma}

One can use the same trick as in \cite{RC09} and extend the Hilbert space $\bigotimes_{i=1}^n \mathcal{H}_i$ to $\bigotimes_{i=1}^n \mathcal{H}_i \otimes \mathcal{H}_i'$ and write the operator $T_n$ as
\begin{equation}
T_n  = \int_{\sum_{i=1}^n x_i^2 + y_i^2 \geq n d_0} {}_{\mathcal{H}'}\langle 0 | U^{\otimes n} \left( |\vec{x}\rangle\!\langle \vec{x}|_\mathbf{S} \otimes |\vec{y}\rangle \! \langle \vec{y}|_\mathbf{\bar{S}} \right) U^{\dagger \otimes n}| 0\rangle_{\mathcal{H}'} \mathrm{d} \vec{x} \mathrm{d}\vec{y}
\end{equation}
where the subscripts $\mathbf{S}$ and $\mathbf{\bar{S}}$ refer to the two possible choices of quadrature (either described by $\bf{s}$ or its complement $\mathbf{\bar{s}}$), $U = e^\pm {\frac{\pi}{4}(a \otimes a'^{\dagger} - a^\dagger \otimes a')}$ is the beamsplitter operator and $|0\rangle_{\mathcal{H}'}$ is the vacuum state on the space  $\bigotimes_{i=1}^n \mathcal{H}_i'$. The $\pm$ sign depends on the specific choice of $s$. 
A possible choice for $\mathbf{S}$ and $\mathbf{\bar{S}}$ would be  $\mathbf{S} = \mathbf{Q} = (Q_1, \cdots, Q_n)$ and $\mathbf{\bar{S}} = \mathbf{P} = (P_1, \cdots, P_n)$. In this section, the blod font is used to describe vectors. 
We also denote $|\alpha=x+iy\rangle$ the coherent state centered in $(x,y)$ in phase space, and will use the equality $\langle 0|U |x\rangle_S |y\rangle_{\bar{S}} = \frac{1}{\sqrt{\pi}} |\alpha\rangle$. 

Let $R$ be the subset of $\mathbbm{R}^{2n}$ corresponding to the support of the integral above:
\begin{equation}
R = \left\{(\vec{x}, \vec{y}) \in \mathbbm{R}^{2n} : \sum_{i=1}^n x_i^2 + y_i^2 \geq n d_0 \right\}.
\end{equation}
For a string $\mathbf{s} \in \{x,y \}^n$, we define the set $\mathcal{R}_s$ as
\begin{equation}
\mathcal{R}_\mathbf{s}  =  \left\{(\vec{x}, \vec{y}) \in \mathbbm{R}^{2n} : \sum_{i=1}^n s_i^2  \geq n \frac{d_0}{2} \right\}.
\end{equation}
Note in particular that the coordinates corresponding to $\bar{\mathbf{S}}$ are unbounded in this set. 
We also introduce the operator $\mathrm{S} = S_1  \oplus \cdots \oplus  S_n$ where $S_i = Q_i (P_i)$ if $s_i = x_i (y_i)$. 
Noting $\bar{s}$ the complement of $s$, one has for all $s$:
\begin{equation}
R \subset \mathcal{R}_\mathbf{s}  \cup \mathcal{R}_{\bar{\mathbf{s}}} 
\end{equation}
which means that 
\begin{equation}
T_n \leq A_\mathbf{s} + A_{\bar{\mathbf{s}}}
\end{equation}
where 
\begin{equation}
A_s  = \int_{ (\vec{x},\vec{y})  \in \mathcal{R}_s } {}_{\mathcal{H}'}\langle 0 | U^{\otimes n} \left( |\vec{x}\rangle\!\langle \vec{x}|_\mathbf{S} \otimes |\vec{y}\rangle \! \langle \vec{y}|_\mathbf{S} \right) U^{\dagger \otimes n}| 0\rangle_{\mathcal{H}'} \mathrm{d} \vec{x} \mathrm{d}\vec{y}.
\end{equation}
Here, we used the fact that the integral of $ |\vec{y}\rangle \! \langle \vec{y}|_\mathbf{S} $ on $\mathbb{R}^n$ is equal to that of $ |\vec{y}\rangle \! \langle \vec{y}|_{\bar{\mathbf{s}}} $: this is simply the $n$-mode generalization of the well-known identity $\int |q\rangle\!\langle q| \mathrm{d}q =\int |p\rangle\!\langle p| \mathrm{d}p$ where $|q\rangle$ and $|p\rangle$ are eigenstates of the quadrature operators, $Q$ and $P$, respectively.
Since the previous relation holds for any string $s$, one has:
\begin{equation}
T_n \leq \frac{1}{2^n} \sum_{s \in \{x,y \}^n}  A_\mathbf{s} + A_{\bar{\mathbf{s}}}
\end{equation}

Let us compute the value of the operator $A_{\mathbf{s}}$:
\begin{align}
A_{\mathbf{s}} &= \int_{ (\vec{x},\vec{y})  \in \mathcal{R}_\mathbf{s} } {}_{\mathcal{H}'}\langle 0 |  \left( \left|\frac{\vec{x}+\vec{y}}{\sqrt{2}}\right\rangle\!\left\langle \frac{\vec{x}+\vec{y}}{\sqrt{2}} \right|_\mathbf{S} \otimes\left|\frac{\vec{x}-\vec{y}}{\sqrt{2}}\right\rangle\!\left\langle \frac{\vec{x}-\vec{y}}{\sqrt{2}} \right|_\mathbf{S} \right) | 0\rangle_{\mathcal{H}'} \mathrm{d} \vec{x} \mathrm{d}\vec{y}\\
 &=\frac{1}{\pi^{n/2}} \int_{ (\vec{x},\vec{y})  \in \mathcal{R}_\mathbf{s} } e^{-\left\|\vec{x}-\vec{y} \right\|^2/2} \left|\frac{\vec{x}+\vec{y}}{\sqrt{2}}\right\rangle\!\left\langle \frac{\vec{x}+\vec{y}}{\sqrt{2}} \right|_\mathbf{S}\mathrm{d} \vec{x} \mathrm{d}\vec{y}.
\end{align}
Changing variables: $\vec{z}_1 = \sqrt{2}\vec{x}, \vec{z}_2 = (\vec{x}+\vec{y})/\sqrt{2}$ (that is, $\vec{x}=\vec{z}_1/\sqrt{2}, \vec{y}=-\vec{z}_1/\sqrt{2} + \sqrt{2} \vec{z}_2$) gives 
\begin{align}
A_{\mathbf{s}} &=\frac{1}{\pi^{n/2}} \int_{ (\vec{z}_1,\vec{z}_2)  \mathrm{s.t.} \, (\vec{x},\vec{y}) \in \mathcal{R}_\mathbf{s}} e^{-\left\|\vec{z}_1-\vec{z}_2 \right\|^2} \left| \vec{z}_2\right\rangle\!\left\langle \vec{z}_2 \right|_\mathbf{S}\mathrm{d} \vec{z}_1 \mathrm{d}\vec{z}_2\\
&=\frac{1}{\pi^{n/2}} \int_{ (\vec{z}_1,\vec{z}_2)  \mathrm{s.t.} \, (\vec{x},\vec{y}) \in\mathcal{R}_\mathbf{s} } e^{-\left\|\vec{z}_1-\mathbf{S} \right\|^2} \left| \vec{z}_2\right\rangle\!\left\langle \vec{z}_2 \right|_\mathbf{S}\mathrm{d} \vec{z}_1 \mathrm{d}\vec{z}_2\\
&=\frac{1}{\pi^{n/2}} \int_{ \|\vec{z}_1\|^2 \geq n d_0} e^{-\left\|\vec{z}_1-\mathbf{S} \right\|^2} \mathrm{d} \vec{z}_1  \int_{\vec{z}_2 \in \mathbb{R}^n} \left| \vec{z}_2\right\rangle\!\left\langle \vec{z}_2 \right|_\mathbf{S}\mathrm{d}\vec{z}_2\\
&= \frac{1}{\pi^{n/2}} \int_{ \|\vec{z}_1\|^2 \geq n d_0} e^{-\left\|\vec{z}_1-\mathbf{S} \right\|^2} \mathrm{d} \vec{z}_1\\
&= F(\mathbf{S}).
\end{align}
We now show for all $a>0$, $F(\mathbf{S}) \leq P^{||\mathbf{S}||^2 \geq a^2} + F(a) \mathbbm{1}$.
To prove it, we need to establish that for any eigenvector $|\vec{s}\rangle$ of the operator $\mathbf{S}$, it holds that 
\begin{equation}
\label{positive}
\langle \vec{s}|F(\mathbf{S})| \vec{s}\rangle \leq \langle \vec{s}| P^{||\mathbf{S}||^2 \geq a^2}| \vec{s}\rangle + \langle \vec{s}| F(a) \mathbbm{1}| \vec{s}\rangle.
\end{equation}

There are two possibilities, 
\begin{itemize}
\item either $||\vec{s}||^2 \geq a^2$, in which case Eq.\ \ref{positive} reads $ F(||\vec{s}||) \leq 1 + F(a)$, which clearly holds, 
\item or $||\vec{s}||^2 \leq a^2$, in which case Eq.\ \ref{positive} reads $F(||\vec{s}||) \leq F(a)$, which holds because $F(x)$ is an increasing function for $x\geq 0$.  
\end{itemize}
Finally, one obtains 
\begin{equation}
T_n \leq \frac{2}{2^n} \sum_{s \in \{x,y \}^n}  P^{Z_n(\mathbf{s}) \geq a^2/n } + 2 F(a) \mathbbm{1}.
\end{equation}
Choosing $a = \sqrt{n d_0/2}$ gives
\begin{equation}
T_n \leq 2 W_n+ 2 F\left(\sqrt{\frac{n d_0}{2}}\right) \mathbbm{1} \leq 2(W_n + e^{-\beta n} \mathbbm{1}),
\end{equation}
which concludes the proof of Lemma \ref{step2ter}.


\end{widetext}


\begin{thebibliography}{24}
\expandafter\ifx\csname natexlab\endcsname\relax\def\natexlab#1{#1}\fi
\expandafter\ifx\csname bibnamefont\endcsname\relax
  \def\bibnamefont#1{#1}\fi
\expandafter\ifx\csname bibfnamefont\endcsname\relax
  \def\bibfnamefont#1{#1}\fi
\expandafter\ifx\csname citenamefont\endcsname\relax
  \def\citenamefont#1{#1}\fi
\expandafter\ifx\csname url\endcsname\relax
  \def\url#1{\texttt{#1}}\fi
\expandafter\ifx\csname urlprefix\endcsname\relax\def\urlprefix{URL }\fi
\providecommand{\bibinfo}[2]{#2}
\providecommand{\eprint}[2][]{\url{#2}}

\bibitem[{\citenamefont{Bennett and Brassard}(1984)}]{bb84}
\bibinfo{author}{\bibfnamefont{C.}~\bibnamefont{Bennett}} \bibnamefont{and}
  \bibinfo{author}{\bibfnamefont{G.}~\bibnamefont{Brassard}}, in
  \emph{\bibinfo{booktitle}{Proceedings of IEEE International Conference on
  Computers, Systems and Signal Processing}} (\bibinfo{year}{1984}), vol.
  \bibinfo{volume}{175}.

\bibitem[{\citenamefont{Ekert}(1991)}]{eke91}
\bibinfo{author}{\bibfnamefont{A.}~\bibnamefont{Ekert}},
  \bibinfo{journal}{Phys. Rev. Lett.} \textbf{\bibinfo{volume}{67}},
  \bibinfo{pages}{661} (\bibinfo{year}{1991}).

\bibitem[{\citenamefont{Scarani et~al.}(2009)\citenamefont{Scarani,
  Bechmann-Pasquinucci, Cerf, Du{\v{s}}ek, L{\"u}tkenhaus, and Peev}}]{SBC08}
\bibinfo{author}{\bibfnamefont{V.}~\bibnamefont{Scarani}},
  \bibinfo{author}{\bibfnamefont{H.}~\bibnamefont{Bechmann-Pasquinucci}},
  \bibinfo{author}{\bibfnamefont{N.}~\bibnamefont{Cerf}},
  \bibinfo{author}{\bibfnamefont{M.}~\bibnamefont{Du{\v{s}}ek}},
  \bibinfo{author}{\bibfnamefont{N.}~\bibnamefont{L{\"u}tkenhaus}},
  \bibnamefont{and} \bibinfo{author}{\bibfnamefont{M.}~\bibnamefont{Peev}},
  \bibinfo{journal}{Rev. Mod. Phys.} \textbf{\bibinfo{volume}{81}},
  \bibinfo{eid}{1301} (\bibinfo{year}{2009}).

\bibitem[{\citenamefont{Renner}(2007)}]{ren07}
\bibinfo{author}{\bibfnamefont{R.}~\bibnamefont{Renner}},
  \bibinfo{journal}{Nat. Phys.} \textbf{\bibinfo{volume}{3}},
  \bibinfo{pages}{645} (\bibinfo{year}{2007}).

\bibitem[{\citenamefont{Christandl et~al.}(2009)\citenamefont{Christandl,
  K\"{o}nig, and Renner}}]{CKR09}
\bibinfo{author}{\bibfnamefont{M.}~\bibnamefont{Christandl}},
  \bibinfo{author}{\bibfnamefont{R.}~\bibnamefont{K\"{o}nig}},
  \bibnamefont{and} \bibinfo{author}{\bibfnamefont{R.}~\bibnamefont{Renner}},
  \bibinfo{journal}{Phys. Rev. Lett.} \textbf{\bibinfo{volume}{102}},
  \bibinfo{eid}{020504} (\bibinfo{year}{2009}).

\bibitem[{\citenamefont{Renner and Cirac}(2009)}]{RC09}
\bibinfo{author}{\bibfnamefont{R.}~\bibnamefont{Renner}} \bibnamefont{and}
  \bibinfo{author}{\bibfnamefont{J.~I.} \bibnamefont{Cirac}},
  \bibinfo{journal}{Phys. Rev. Lett.} \textbf{\bibinfo{volume}{102}},
  \bibinfo{pages}{110504} (\bibinfo{year}{2009}).

\bibitem[{\citenamefont{Weedbrook et~al.}(2012)\citenamefont{Weedbrook,
  Pirandola, Garc{\'i}a-Patr\'on, Cerf, Ralph, Shapiro, and Lloyd}}]{WPG12}
\bibinfo{author}{\bibfnamefont{C.}~\bibnamefont{Weedbrook}},
  \bibinfo{author}{\bibfnamefont{S.}~\bibnamefont{Pirandola}},
  \bibinfo{author}{\bibfnamefont{R.}~\bibnamefont{Garc{\'i}a-Patr\'on}},
  \bibinfo{author}{\bibfnamefont{N.~J.} \bibnamefont{Cerf}},
  \bibinfo{author}{\bibfnamefont{T.~C.} \bibnamefont{Ralph}},
  \bibinfo{author}{\bibfnamefont{J.~H.} \bibnamefont{Shapiro}},
  \bibnamefont{and} \bibinfo{author}{\bibfnamefont{S.}~\bibnamefont{Lloyd}},
  \bibinfo{journal}{Rev. Mod. Phys.} \textbf{\bibinfo{volume}{84}},
  \bibinfo{pages}{621} (\bibinfo{year}{2012}).

\bibitem[{\citenamefont{Qi et~al.}(2010)\citenamefont{Qi, Zhu, Qian, and
  Lo}}]{QZQ10}
\bibinfo{author}{\bibfnamefont{B.}~\bibnamefont{Qi}},
  \bibinfo{author}{\bibfnamefont{W.}~\bibnamefont{Zhu}},
  \bibinfo{author}{\bibfnamefont{L.}~\bibnamefont{Qian}}, \bibnamefont{and}
  \bibinfo{author}{\bibfnamefont{H.}~\bibnamefont{Lo}}, \bibinfo{journal}{New
  J. Phys.} \textbf{\bibinfo{volume}{12}}, \bibinfo{pages}{103042}
  (\bibinfo{year}{2010}).

\bibitem[{\citenamefont{Jouguet et~al.}(2012)\citenamefont{Jouguet,
  Kunz-Jacques, Leverrier, Grangier, and Diamanti}}]{JKL12}
\bibinfo{author}{\bibfnamefont{P.}~\bibnamefont{Jouguet}},
  \bibinfo{author}{\bibfnamefont{S.}~\bibnamefont{Kunz-Jacques}},
  \bibinfo{author}{\bibfnamefont{A.}~\bibnamefont{Leverrier}},
  \bibinfo{author}{\bibfnamefont{P.}~\bibnamefont{Grangier}}, \bibnamefont{and}
  \bibinfo{author}{\bibfnamefont{E.}~\bibnamefont{Diamanti}},
  \bibinfo{journal}{submitted to QCRYPT 2012}  (\bibinfo{year}{2012}).

\bibitem[{\citenamefont{Garc\'{\i}a-Patr\'{o}n and Cerf}(2006)}]{GC06}
\bibinfo{author}{\bibfnamefont{R.}~\bibnamefont{Garc\'{\i}a-Patr\'{o}n}}
  \bibnamefont{and} \bibinfo{author}{\bibfnamefont{N.~J.} \bibnamefont{Cerf}},
  \bibinfo{journal}{Phys. Rev. Lett.} \textbf{\bibinfo{volume}{97}},
  \bibinfo{pages}{190503} (\bibinfo{year}{2006}).

\bibitem[{\citenamefont{Navascu\'{e}s et~al.}(2006)\citenamefont{Navascu\'{e}s,
  Grosshans, and Ac\'{\i}n}}]{NGA06}
\bibinfo{author}{\bibfnamefont{M.}~\bibnamefont{Navascu\'{e}s}},
  \bibinfo{author}{\bibfnamefont{F.}~\bibnamefont{Grosshans}},
  \bibnamefont{and}
  \bibinfo{author}{\bibfnamefont{A.}~\bibnamefont{Ac\'{\i}n}},
  \bibinfo{journal}{Phys. Rev. Lett.} \textbf{\bibinfo{volume}{97}},
  \bibinfo{pages}{190502} (\bibinfo{year}{2006}).

\bibitem[{\citenamefont{Furrer et~al.}(2011)\citenamefont{Furrer, Franz, Berta,
  Leverrier, Scholz, Tomamichel, and Werner}}]{FFB11}
\bibinfo{author}{\bibfnamefont{F.}~\bibnamefont{Furrer}},
  \bibinfo{author}{\bibfnamefont{T.}~\bibnamefont{Franz}},
  \bibinfo{author}{\bibfnamefont{M.}~\bibnamefont{Berta}},
  \bibinfo{author}{\bibfnamefont{A.}~\bibnamefont{Leverrier}},
  \bibinfo{author}{\bibfnamefont{V.}~\bibnamefont{Scholz}},
  \bibinfo{author}{\bibfnamefont{M.}~\bibnamefont{Tomamichel}},
  \bibnamefont{and} \bibinfo{author}{\bibfnamefont{R.}~\bibnamefont{Werner}},
  \bibinfo{journal}{Arxiv preprint ArXiv:1112.2179}  (\bibinfo{year}{2011}).

\bibitem[{\citenamefont{Sheridan et~al.}(2010)\citenamefont{Sheridan, Le, and
  Scarani}}]{SLS10}
\bibinfo{author}{\bibfnamefont{L.}~\bibnamefont{Sheridan}},
  \bibinfo{author}{\bibfnamefont{T.}~\bibnamefont{Le}}, \bibnamefont{and}
  \bibinfo{author}{\bibfnamefont{V.}~\bibnamefont{Scarani}},
  \bibinfo{journal}{New J. Phys.} \textbf{\bibinfo{volume}{12}},
  \bibinfo{pages}{123019} (\bibinfo{year}{2010}).

\bibitem[{\citenamefont{M{\"u}ller-Quade and Renner}(2009)}]{MKR09}
\bibinfo{author}{\bibfnamefont{J.}~\bibnamefont{M{\"u}ller-Quade}}
  \bibnamefont{and} \bibinfo{author}{\bibfnamefont{R.}~\bibnamefont{Renner}},
  \bibinfo{journal}{New J. Phys.} \textbf{\bibinfo{volume}{11}},
  \bibinfo{pages}{085006} (\bibinfo{year}{2009}).

\bibitem[{\citenamefont{Weedbrook et~al.}(2004)\citenamefont{Weedbrook, Lance,
  Bowen, Symul, Ralph, and Lam}}]{WLB04}
\bibinfo{author}{\bibfnamefont{C.}~\bibnamefont{Weedbrook}},
  \bibinfo{author}{\bibfnamefont{A.~M.} \bibnamefont{Lance}},
  \bibinfo{author}{\bibfnamefont{W.~P.} \bibnamefont{Bowen}},
  \bibinfo{author}{\bibfnamefont{T.}~\bibnamefont{Symul}},
  \bibinfo{author}{\bibfnamefont{T.~C.} \bibnamefont{Ralph}}, \bibnamefont{and}
  \bibinfo{author}{\bibfnamefont{P.~K.} \bibnamefont{Lam}},
  \bibinfo{journal}{Phys. Rev. Lett.} \textbf{\bibinfo{volume}{93}},
  \bibinfo{pages}{170504} (\bibinfo{year}{2004}).

\bibitem[{\citenamefont{Jouguet et~al.}(2011)\citenamefont{Jouguet,
  Kunz-Jacques, and Leverrier}}]{JKL11}
\bibinfo{author}{\bibfnamefont{P.}~\bibnamefont{Jouguet}},
  \bibinfo{author}{\bibfnamefont{S.}~\bibnamefont{Kunz-Jacques}},
  \bibnamefont{and}
  \bibinfo{author}{\bibfnamefont{A.}~\bibnamefont{Leverrier}},
  \bibinfo{journal}{Phys. Rev. A} \textbf{\bibinfo{volume}{84}},
  \bibinfo{pages}{062317} (\bibinfo{year}{2011}).

\bibitem[{\citenamefont{Renner}(2005)}]{ren05}
\bibinfo{author}{\bibfnamefont{R.}~\bibnamefont{Renner}}, Ph.D. thesis,
  \bibinfo{school}{ETH Zurich} (\bibinfo{year}{2005}),
  \bibinfo{note}{http://arxiv.org/abs/quant-ph/0512258}.

\bibitem[{\citenamefont{Leverrier and Grangier}(2010)}]{LG10b}
\bibinfo{author}{\bibfnamefont{A.}~\bibnamefont{Leverrier}} \bibnamefont{and}
  \bibinfo{author}{\bibfnamefont{P.}~\bibnamefont{Grangier}},
  \bibinfo{journal}{Phys. Rev. A} \textbf{\bibinfo{volume}{81}},
  \bibinfo{pages}{062314} (\bibinfo{year}{2010}).

\bibitem[{\citenamefont{Leverrier et~al.}(2010)\citenamefont{Leverrier,
  Grosshans, and Grangier}}]{LGG10}
\bibinfo{author}{\bibfnamefont{A.}~\bibnamefont{Leverrier}},
  \bibinfo{author}{\bibfnamefont{F.}~\bibnamefont{Grosshans}},
  \bibnamefont{and} \bibinfo{author}{\bibfnamefont{P.}~\bibnamefont{Grangier}},
  \bibinfo{journal}{Phys. Rev. A} \textbf{\bibinfo{volume}{81}},
  \bibinfo{pages}{062343} (\bibinfo{year}{2010}).

\bibitem[{\citenamefont{Leverrier et~al.}(2009)\citenamefont{Leverrier, Karpov,
  Grangier, and Cerf}}]{LKG09}
\bibinfo{author}{\bibfnamefont{A.}~\bibnamefont{Leverrier}},
  \bibinfo{author}{\bibfnamefont{E.}~\bibnamefont{Karpov}},
  \bibinfo{author}{\bibfnamefont{P.}~\bibnamefont{Grangier}}, \bibnamefont{and}
  \bibinfo{author}{\bibfnamefont{N.~J.} \bibnamefont{Cerf}},
  \bibinfo{journal}{New J. Phys.} \textbf{\bibinfo{volume}{11}},
  \bibinfo{pages}{115009} (\bibinfo{year}{2009}).

\bibitem[{\citenamefont{Leverrier and Cerf}(2009)}]{LC09}
\bibinfo{author}{\bibfnamefont{A.}~\bibnamefont{Leverrier}} \bibnamefont{and}
  \bibinfo{author}{\bibfnamefont{N.~J.} \bibnamefont{Cerf}},
  \bibinfo{journal}{Phys. Rev. A} \textbf{\bibinfo{volume}{80}},
  \bibinfo{eid}{010102} (\bibinfo{year}{2009}).

\bibitem[{\citenamefont{Leverrier}(2012)}]{lev12}
\bibinfo{author}{\bibfnamefont{A.}~\bibnamefont{Leverrier}},
  \bibinfo{journal}{Phys. Rev. A} \textbf{\bibinfo{volume}{85}},
  \bibinfo{pages}{022339} (\bibinfo{year}{2012}).

\bibitem[{\citenamefont{Laurent and Massart}(2000)}]{LM00}
\bibinfo{author}{\bibfnamefont{B.}~\bibnamefont{Laurent}} \bibnamefont{and}
  \bibinfo{author}{\bibfnamefont{P.}~\bibnamefont{Massart}},
  \bibinfo{journal}{The Annals of Statistics} \textbf{\bibinfo{volume}{28}},
  \bibinfo{pages}{1302} (\bibinfo{year}{2000}).

\bibitem[{\citenamefont{Arvind et~al.}(1995)\citenamefont{Arvind, Dutta,
  Mukunda, and Simon}}]{ADMS95}
\bibinfo{author}{\bibnamefont{Arvind}},
  \bibinfo{author}{\bibfnamefont{B.}~\bibnamefont{Dutta}},
  \bibinfo{author}{\bibfnamefont{N.}~\bibnamefont{Mukunda}}, \bibnamefont{and}
  \bibinfo{author}{\bibfnamefont{R.}~\bibnamefont{Simon}},
  \bibinfo{journal}{Pramana} \textbf{\bibinfo{volume}{45}},
  \bibinfo{pages}{471} (\bibinfo{year}{1995}).

\end{thebibliography}
\end{document}